\documentclass{amsart}
\usepackage{amsfonts}
\usepackage{amsmath}
\usepackage{amssymb}

\theoremstyle{plain}
\newtheorem{theorem}{Theorem}[section]
\newtheorem{lemma}[theorem]{Lemma}
\newtheorem{proposition}[theorem]{Proposition}
\theoremstyle{definition}

\theoremstyle{remark}
\newtheorem{remark}[theorem]{Remark}
\numberwithin{equation}{section}

\begin{document}
\title{Differential-algebraic integrability analysis of the generalized
Riemann type and Korteweg-de Vries hydrodynamical equations}
\author{Anatoliy K. Prykarpatsky*), Orest D. Artemovych**), \\ Ziemowit Popowicz***), Maxim V. Pavlov$\ddagger $)}
\address{*) The Department of Mining Geodesy at the AGH University of
Science and Technology, Crac\'{o}w 30059, Poland, and the Ivan Franko State
Pedagogical University, Drohobych, Lviv region, Ukraine; pryk.anat@ua.fm}
\address{**) The Department of Algebra and Topology at the Faculty of
Mathematics and Informatics of the Vasyl Stefanyk Pre-Carpathian National
University, Ivano-Frankivsk, Ukraine, and the Institute of Mathematics and
Informatics at the Tadeusz Kosciuszko University of Technology, Crac\'{o}w,
Poland; artemo@usk.pk.edu.pl}
\address{***) The Institute of Theoretical Physics, University of Wroc\l aw,
Poland; ziemek@ift.uni.wroc.pl }
\address{$\ddagger $) The Department of Mathematical Physics, P.N. Lebedev
Physical Institute, 53 Leninskij Prospekt, Moscow 119991, Russia;
M.V.Pavlov@lboro.ac.uk }
\maketitle

\begin{abstract}
A differential-algebraic approach to studying the Lax type integrability of
the generalized Riemann type hydrodynamic equations at $N=3,4$ is devised.
The approach is also applied to studying the Lax type integrability of the
well known Korteweg-de Vries dynamical system.
\end{abstract}

\section{Introduction}

Nonlinear hydrodynamic equations are \ of constant interest since the
classical works by B. Riemann in the general three-dimensional case, having
paid special attention to their one-dimensional spatial reduction, for which
he devised the generalized method of characteristics and Riemann invariants.
\ These methods appeared to be very effective \cite{Wh,PBB} in investigating
many types of nonlinear spatially one-dimensional systems of hydrodynamical
type and, in particular, the characteristics method in the form of a
"reciprocal" transformation of variables has been \ used recently in
studying the so called Gurevich-Zybin system \cite{GZ,GZ1} in \cite{Pav} and
the Whitham type system in \cite{PrPryt,GBPPP,Sak}. Moreover, this method
was further effectively applied to studying solutions to a generalized \cite%
{GPPP} (owing to D. Holm and M. Pavlov) Riemann type hydrodynamical system
\begin{equation}
D_{t}^{N}u=0,\text{ \ \ }D_{t}:=\partial /\partial t+u\partial /\partial x,%
\text{ \ }N\in \mathbb{Z}_{+},  \label{H0a}
\end{equation}%
where $u\in C^{\infty }(\mathbb{R}^{2};\mathbb{R})$ is a smooth function.
The case $N=2$ was recently analyzed in detail in \cite{GBPPP,GPPP} making
use of the standard symplectic theory techniques. In particular, there was
demonstrated that the Riemann type hydrodynamical system (\ref{H0a}) at $%
N=2, $ looking upon putting $z:=D_{t}u$ equivalently as
\begin{equation}
\left.
\begin{array}{c}
u_{t}=z-uu_{x} \\
z_{t}=-uz_{x}%
\end{array}%
\right\} ,  \label{H0b}
\end{equation}%
allows the following Lax type representation
\begin{equation}
\begin{array}{c}
\partial f/\partial x=\ell \lbrack u,z;\lambda ]f,\text{ \ \ }\partial
f/\partial t=p(\ell )f,\text{ \ \ \ \ }p(\ell ):=-u\ell \lbrack u,z;\lambda
]+q(\lambda ), \\
\text{\ }\ell \lbrack u,z;\lambda ]:=\left(
\begin{array}{cc}
-\lambda u_{x} & -z_{x} \\
2\lambda ^{2} & \lambda u_{x}%
\end{array}%
\right) ,\text{ \ }q(\lambda ):=\left(
\begin{array}{cc}
0 & 0 \\
-\lambda & 0%
\end{array}%
\right) ,\text{\ } \\
p(\ell )=\left(
\begin{array}{cc}
\lambda u_{x}u & z_{x}u \\
-\lambda -2\lambda ^{2}u & -\lambda u_{x}u%
\end{array}%
\right) ,%
\end{array}
\label{H0c}
\end{equation}%
where $f\in C^{(\infty )}(\mathbb{R}^{2};\mathbb{C}^{2})$ \ and $\lambda \in
\mathbb{C}$ is an arbitrary spectral parameter. \ Making use of a method
devised in \cite{PM,MBPS,HPP} and based on the spectral theory and related
very complicated symplectic theory relationships in \cite{GPPP,GBPPP,PoP}
the corresponding Lax type representations for the cases $N=3,4$ were
constructed in explicit form.

In this work a new \ and very simple differential-algebraic approach to
studying the Lax type integrability of the generalized Riemann type
hydrodynamic equations at $N=3,4$ is devised. It can be easily generalized
for treating the problem for arbitrary integers $N\in \mathbb{Z}_{+}.$ The
approach is also applied to studying the Lax type integrability of the well
known Korteweg-de Vries dynamical system.

\section{ The differential-algebraic description of the Lax type
integrability of the generalized Riemann type hydrodynamical equation at N=3
and 4}

\subsection{The differential-algebraic preliminaries}

Take the ring $\mathcal{K}:=\mathbb{R}\{\{x,t\}\},$ $(x,t)\in \mathbb{R}%
^{2}, $ of convergent germs of real-valued smooth functions from $C^{(\infty
)}(\mathbb{R}^{2};\mathbb{R})$ and construct \cite{Rit,Ko,Wei,CH} the
associated differential polynomial ring $\mathcal{K}\{u\}:=\mathcal{K}%
[\Theta u]$ with respect to a functional variable $u,$ where $\Theta $
denotes the standard monoid of all operators generated by commuting
differentiations $\partial /\partial x:=D_{x}$ and $\partial /\partial t.$
The ideal $I\mathcal{\{}u\mathcal{\}}\subset \mathcal{K}\{u\}$ is called
\cite{Rit,Ko} differential if the condition $I\mathcal{\{}u\mathcal{\}}%
=\Theta I\mathcal{\{}u\mathcal{\}}$ holds.

Consider now the additional differentiation
\begin{equation}
D_{t}:\mathcal{K}\{u\}\rightarrow \mathcal{K}\{u\},  \label{A0}
\end{equation}%
depending on the functional variable $u,$ which satisfies the Lie-algebraic
commutator condition
\begin{equation}
\lbrack D_{x},D_{t}]=(D_{x}u)D_{x},  \label{A1}
\end{equation}%
for all $(x,t)\in \mathbb{R}^{2}.$ As a simple consequence of (\ref{A1}) the
following general (suitably normalized) \textit{representation }of the
differentiation (\ref{A0}) \
\begin{equation}
D_{t}=\partial /\partial t+u\partial /\partial x  \label{A1-0}
\end{equation}%
in the differential ring $\ \mathcal{K}\{u\}$ holds. \ Impose now on the
differentiation (\ref{A0}) a new algebraic constraint
\begin{equation}
D_{t}^{N}u=0,  \label{A1-1}
\end{equation}%
defining some smooth functional set (or "\textit{manifold}") $\mathcal{M}%
^{(N)}$ of functions $u\in \mathbb{R}\{\{x,t\}\},$ and which allows to
reduce naturally the initial ring $\mathcal{K}\{u\}$ to the basic ring $%
\mathcal{K}\{u\}|_{\mathcal{M}_{(N)}}\subseteq \mathbb{R}\{\{x,t\}\}.$ In
this case the following natural problem of constructing the corresponding
representation of differentiation (\ref{A0}) arises: \textit{to find an
equivalent linear representation of the reduced differentiation }$D_{t}|_{%
\mathcal{M}_{(N)}}:\mathbb{R}^{p(N)}\{\{x,t\}\}\rightarrow \mathbb{R}%
^{p(N)}\{\{x,t\}\}$\textit{\ in the functional vector space }$\mathbb{R}%
^{p(N)}\{\{x,t\}\}$\textit{\ for some specially chosen integer dimension }$%
p(N)\in \mathbb{Z}_{+}.$

As it will be shown below for the cases $N=3$ and $N=4,$ this problem is
completely analytically solvable, giving rise to the corresponding Lax type
integrability of the generalized Riemann type hydrodynamical system (\ref%
{H0a}). Moreover, the same problem is also solvable for the more complicated
constraint
\begin{equation}
D_{t}u-D_{x}^{3}u=0,  \label{A1-2}
\end{equation}%
equivalent to the well known Lax type integrable nonlinear Korteweg-de Vries
dynamical system.

\subsection{The generalized Riemann type hydrodynamical equation: the case
N=3}

To proceed with analyzing the above formulated representation problem for
the generalized Riemann type equation (\ref{A1-1}) at $N=3,$ we first
construct an adjoint to $\ $the differential ring $\mathcal{K}\{u\}$ and
invariant with respect to differentiation (\ref{A1-0}) so called "Riemann
differential ideal" $\ \ R\mathcal{\{}u\mathcal{\}}\subset \mathcal{K}\{u\}$
as
\begin{eqnarray}
R\mathcal{\{}u\mathcal{\}} &:&=\{\lambda \sum_{n\in \mathbb{Z}%
_{+}}f_{n}^{(1)}D_{x}^{n}u-\sum_{n\in \mathbb{Z}%
_{+}}f_{n}^{(2)}D_{t}D_{x}^{n}u+\sum_{n\in \mathbb{Z}%
_{+}}f_{n}^{(3)}D_{t}^{2}D_{x}^{n}u:D_{t}^{3}u=0,\text{ \ }  \notag \\
f_{n}^{(k)} &\in &\mathcal{K}\{u\},k=\overline{1,3},n\in \mathbb{Z}%
_{+}\}\subset \mathcal{K}\{u\},  \label{A2a}
\end{eqnarray}%
where $\lambda \in \mathbb{R}$ is an arbitrary parameter, and formulate the
following simple but important lemma.

\begin{lemma}
\label{Lm_1}The kernel $Ker$ $D_{t}\subset {R}\mathcal{\{}u\mathcal{\}}$ of
the differentiation $D_{t}:\mathcal{K}\{u\}\rightarrow \mathcal{K}\{u\},$
reduced modulo the Riemann differential ideal $R\mathcal{\{}u\mathcal{\}}$ $%
\subset \mathcal{K}\{u\},$ is generated by elements satisfying the following
linear functional-differential relationships:%
\begin{equation}
D_{t}f^{(1)}=0,\text{ \ \ }D_{t}f^{(2)}=\lambda f^{(1)},\text{ \ }%
D_{t}f^{(3)}=f^{(2)},  \label{A3}
\end{equation}%
where, by definition, $\ f^{(k)}:=f^{(k)}(\lambda )=\sum_{n\in \mathbb{Z}%
_{+}}f_{n}^{(k)}\lambda ^{n}\in \mathcal{K}\{u\}|_{\mathcal{M}_{(3)}}=%
\mathbb{R}{\{\{}x,t\}\},$ $k=\overline{1,3},$ and $\lambda \in \mathbb{R}$
is arbitrary.
\end{lemma}

It is easy to see that equations (\ref{A3}) can be equivalently rewritten
both in the matrix form as
\begin{equation}
D_{t}f=q(\lambda )f,\text{ \ \ }q(\lambda ):=\left(
\begin{array}{ccc}
0 & 0 & 0 \\
\lambda & 0 & 0 \\
0 & 1 & 0%
\end{array}%
\right) ,  \label{A4}
\end{equation}%
where $f:=(f^{(1)},f^{(2)},f^{(3)})^{\intercal }\in \mathcal{K}^{3}\{u\}|_{%
\mathcal{M}_{(3)}},$ $\lambda \in \mathbb{R}$ is an arbitrary "spectral"
parameter, and in the compact scalar form as
\begin{equation}
D_{t}^{3}f_{3}=0  \label{A5}
\end{equation}%
for an element $f_{3}\in \mathcal{K}\{u\}|_{\mathcal{M}_{(3)}}.$ Here it is
worth to \ note that the Riemann differential ideal (\ref{A2a}), satisfying
the $D_{t}$-invariance condition, is in this case maximal. Now we can
construct by means of relationship (\ref{A5}) a new invariant, the so-called
"Lax differential ideal" $\ L\mathcal{\{}u\mathcal{\}}\subset \mathcal{K}%
\{u\},$ isomorphic to the Riemann differential ideal $R\mathcal{\{}u\mathcal{%
\}}\subset \mathcal{K}\{u\}$ and realizing the Lax type integrability
condition of the Riemann type hydrodynamical equation (\ref{H0a}). Namely,
based on the result of Lemma \ref{Lm_1} the following proposition holds.

\begin{proposition}
The expression (\ref{A4}) is an adjoint linear matrix representation in the
space $\mathbb{R}^{3}\mathbb{\{\{}x,t\}\}$ of the differentiation $D_{t}:%
\mathcal{K}\{u\}\rightarrow \mathcal{K}\{u\},$ reduced to the ideal $R%
\mathcal{\{}u\mathcal{\}}\subset \mathcal{K}\{u\}.$ The related $D_{x}$- and
$D_{t}$-invariant Lax differential ideal $L\mathcal{\{}u\mathcal{\}}\subset
\mathcal{K}\{u\},\ $which is isomorphic to the invariant Riemann
differential ideal $R\mathcal{\{}u\mathcal{\}}\subset \mathcal{K}\{u\},$ is
generated by the element $f_{3}(\lambda )\in \mathcal{K}\{u\},\lambda \in
\mathbb{R},$ satisfying condition (\ref{A5}), and equals%
\begin{eqnarray}
L\mathcal{\{}u\mathcal{\}} &:&=\{g_{1}f_{3}(\lambda
)+g_{2}D_{t}f_{3}(\lambda )+g_{3}D_{t}^{2}f_{3}(\lambda
):D_{t}^{3}f_{3}(\lambda )=0,  \notag \\
\lambda &\in &\mathbb{R},g_{j}\in \mathcal{K}\{u\},j=\overline{1,3}\}\subset
\mathcal{K}\{u\}.  \label{A6}
\end{eqnarray}
\end{proposition}

We now construct a related adjoint linear matrix representation in the
functional vector space $\mathbb{R}^{3}{\{\{}x,t{\}\}}$ for the
differentiation $D_{x}:\mathcal{K}\{u\}\rightarrow \mathcal{K}\{u\},$
reduced modulo the Lax differential ideal $L\mathcal{\{}u\mathcal{\}}\subset
\mathcal{K}\{u\}.$ For this problem to be solved, we need to take into
account the commutator relationship (\ref{A1}) and the important invariance
condition of the Lax differential ideal $L\mathcal{\{}u\mathcal{\}}\subset
\mathcal{K}\{u\}$ with respect to the differentiation $D_{x}:\mathcal{K}%
\{u\}\rightarrow \mathcal{K}\{u\}.$ As a result of simple but slightly
tedious calculations one obtains the following matrix representation:%
\begin{equation}
D_{x}f=\ell \lbrack u,v,z;\lambda ]f,\text{ \ }\ell \lbrack u,v,z;\lambda
]:=\left(
\begin{array}{ccc}
\lambda u_{x} & -v_{x} & z_{x} \\
3\lambda ^{2} & -2\lambda u_{x} & \lambda v_{x} \\
6\lambda ^{2}r[u,v,z] & -3\lambda & \lambda u_{x}%
\end{array}%
\right) ,  \label{A7}
\end{equation}%
where, by definition, $v:=D_{t}u,z:=D_{t}v,$ $(...)_{x}:=D_{x}(...),$ a
vector $f\in \mathbb{R}^{3}\{\{x,t\}\},$ $\lambda \in \mathbb{R}$ is an
arbitrary spectral parameter and a smooth functional mapping $r:\mathcal{%
\tilde{M}}_{(3)}\rightarrow \mathbb{R}\{\{x,t\}\},\mathcal{\tilde{M}}%
_{(3)}:=\sqcap _{j=1}^{3}D_{t}^{j}\mathcal{M}_{(3)},$ solves the following
functional-differential equation
\begin{equation}
D_{t}r+rD_{x}u=1.  \label{A8}
\end{equation}%
Moreover, the matrix $\ell :=\ell \lbrack u,v,z;\lambda ]:\mathbb{R}^{3}{\{\{%
}x,t\}\}\mathbb{\rightarrow R}^{3}{\{\{}x,t\}\}$ satisfies the following
determining functional-differential equation:%
\begin{equation}
D_{t}\ell +\ell D_{x}u=[q(\lambda ),\ell ],  \label{A8a}
\end{equation}%
where $[\cdot ,\cdot ]$ denotes the usual matrix commutator in the
functional space $\mathbb{R}^{3}{\{\{}x,t\}\}.$ The following proposition
solving the representation problem posed above, holds.

\begin{proposition}
The expression (\ref{A7}) is an adjoint linear matrix representation in the
space $\mathbb{R}^{3}{\{\{}x,t\}\}$ of the differentiation $D_{x}:\mathcal{K}%
\{u\}\rightarrow \mathcal{K}\{u\},$ reduced \ modulo the invariant Lax
differential ideal $L\mathcal{\{}u\mathcal{\}}\subset \mathcal{K}\{u\},$
given by (\ref{A6}).
\end{proposition}

\begin{remark}
Here it is necessary to mention that the matrix representation (\ref{A4})
coincides completely with that obtained before in the work \cite{GPPP} by
means of completely different methods, based mainly on the
gradient-holonomic algorithm, devised in \cite{PM,MBPS,HPP}. The presented
derivation of these representations (\ref{A4}) and (\ref{A7}) is much easier
and simpler that can be explained by a deeper insight into the integrability
problem, devised above using the differential algebraic approach.
\end{remark}

To proceed further, it is now worth to observing that the invariance
condition for the Lax differential ideal $L\mathcal{\{}u\mathcal{\}}\subset
\mathcal{K}\{u\}$ with respect to the differentiations $D_{x},D_{t}:\mathcal{%
K}\{u\}\rightarrow \mathcal{K}\{u\}$ is also equivalent to the related Lax
type representation for the generalized Riemann type equation \ref{H0a} in
the following dynamical system form:
\begin{equation}
\left.
\begin{array}{c}
u_{t}=v-uu_{x} \\
v_{t}=z-uv_{x} \\
z_{t}=-uz_{x}%
\end{array}%
\right\} :=K[u,v,z],  \label{A9}
\end{equation}%
Namely, the following theorem, summing up the results obtained above, holds.

\begin{theorem}
The linear differential-matrix expressions (\ref{A4}) and (\ref{A7}) in the
space $\mathbb{R}^{3}{\{\{}x,t\}\}$ for differentiations $D_{t}:\mathcal{K}%
\{u\}\rightarrow $ $\mathcal{K}\{u\}$ and $D_{x}:\mathcal{K}\{u\}\rightarrow
\mathcal{K}\{u\},$ respectively, provide us with the standard Lax type
representation for the generalized Riemann type equation (\ref{H0a}) in the
equivalent dynamical system form (\ref{A9}), thereby implying its Lax type
integrability.
\end{theorem}

The next problem of great interest is to construct, making use of the
differential-algebraic tools, the functional-differential solutions to the
determining equation (\ref{R3}), and to construct the corresponding
differential-algebraic analogs of the symplectic structures characterizing
the differentiations $D_{x},D_{t}:\mathcal{K}\{u\}\rightarrow \mathcal{K}%
\{u\},$ as well as the local densities of the related conservation laws,
which were derived in \cite{GPPP,PoP}.

\subsection{The solution set analysis of the functional-differential
equation $D_{t}r+rD_{x}u=1$}

We consider the generalized Riemann type dynamical system (\ref{A9}) on a $\
$suitable $2\pi $-periodic functional manifold $\mathcal{M}_{(3)}\subset
\mathbb{R}^{3}{\{\{}x,t\}\}:$
\begin{equation}
\left.
\begin{array}{c}
u_{t}=v-uu_{x} \\
v_{t}=z-uv_{x} \\
z_{t}=-uz_{x}%
\end{array}%
\right\} :=K[u,v,z],  \label{R1}
\end{equation}%
which, as shown above and in \cite{GPPP,PoP}, possesses the following Lax
type representation:
\begin{equation}
\begin{array}{c}
f_{x}=\ell \lbrack u,v,z;\lambda ]{f},\text{ \ \ }f_{t}=p(\ell )f,\text{ \ \
}p(\ell ):=-u\ell \lbrack u,v,z;\lambda ]+q(\lambda ), \\
\\
\\
\ell \lbrack u,v,z;\lambda ]=\left(
\begin{array}{ccc}
\lambda u_{x} & -v_{x} & z_{x} \\
3\lambda ^{2} & -2\lambda u_{x} & \lambda v_{x} \\
6\lambda ^{2}r[u,v,z] & -3\lambda & \lambda u_{x}%
\end{array}%
\right) ,\text{ \ }q(\lambda ):=\left(
\begin{array}{ccc}
0 & 0 & 0 \\
\lambda & 0 & 0 \\
0 & 1 & 0%
\end{array}%
\right) , \\
\\
\\
\text{ \ }p(\ell )=\left(
\begin{array}{ccc}
-\lambda uu_{x} & uv_{x} & -uz_{x} \\
-3u\lambda ^{2}+\lambda & 2\lambda uu_{x} & -\lambda uv_{x} \\
-6\lambda ^{2}r[u,v,z]u & 1+3u\lambda & -\lambda uu_{x}%
\end{array}%
\right) ,%
\end{array}%
,  \label{R2}
\end{equation}%
where $f\in L_{\infty }(\mathbb{R};\mathbb{E}^{3}),$ $\lambda \in \mathbb{R}$
is an arbitrary spectral parameter and a function $r:\mathcal{M}\rightarrow
\mathbb{R}$ satisfies the following functional-differential equation:%
\begin{equation}
D_{t}r+rD_{x}u=1  \label{R3}
\end{equation}%
under the commutator condition (\ref{A1}).

Below we will describe all functional solutions to equation (\ref{R3}),
making use of the lemma in \cite{GBPPP,PoP}.

\begin{lemma}
The following functions
\begin{equation}
B_{0}=\xi (z),\text{ }B_{1}=u-tv+zt^{2}/2,\text{ }B_{2}=v-zt,\text{ \ }%
B_{3}=x-tu+vt^{2}/2-zt^{3}/6,  \label{R3a}
\end{equation}%
where $\xi :D_{t}^{2}\mathcal{M}\rightarrow \mathbb{R}\{\{x,t\}\}$ is an
arbitrary smooth mapping, are the main invariants of the Riemann type
dynamical system (\ref{R1}), satisfying the determining condition
\begin{equation}
D_{t}B=0.  \label{R4}
\end{equation}
\end{lemma}

As a simple consequence of relationships (\ref{R3a}) the next lemma holds.

\begin{lemma}
The local functionals
\begin{equation}
b_{0}:=\xi (z),b_{1}:=\frac{u}{z}-\frac{v^{2}}{2z^{2}},b_{2}:=\frac{u_{x}}{%
z_{x}}-\frac{v_{x}^{2}}{2z_{x}^{2}},b_{3}:=x-\frac{uv}{z}+\frac{v^{3}}{3z^{2}%
}  \label{R5}
\end{equation}%
and
\begin{equation*}
\tilde{b}_{1}:=\frac{v}{z},\tilde{b}_{2}:=\frac{v_{x}}{z_{x}}
\end{equation*}%
on the functional manifold $\mathcal{\tilde{M}}_{(3)}$ are the basic
functional solutions $b_{j}:\mathcal{\tilde{M}}_{(3)}$ $\rightarrow \mathbb{R%
}\{\{x,t\}\},$ $j=\overline{0,3},$ and $\tilde{b}_{k}:\mathcal{\tilde{M}}%
_{(3)}\rightarrow \mathbb{R}\{\{x,t\}\},k=\overline{1,2},$ to the
determining functional-differential equations%
\begin{equation}
D_{t}b=0  \label{R6}
\end{equation}%
and
\begin{equation}
D_{t}\tilde{b}=1,  \label{R7}
\end{equation}%
respectively.
\end{lemma}

Now one can formulate the following theorem about the general solution set
to the functional-differential equation (\ref{R6}).

\begin{theorem}
The following infinite hierarchies
\begin{equation}
\eta _{1,j}^{(n)}:=(\alpha D_{x})^{n}b_{j},\text{ \ }\eta
_{2,k}^{(n)}:=(\alpha D_{x})^{n+1}\tilde{b}_{k},  \label{R8}
\end{equation}%
where $\alpha :=1/z_{x},$ $j=\overline{0,3},$ $k=\overline{1,2}$ and $n\in
\mathbb{Z}_{+},$ are the basic functional solutions to the
functional-differential equation (\ref{R6}), that is
\begin{equation}
D_{t}\eta _{s,j}^{(n)}=0  \label{R9}
\end{equation}%
for $s=\overline{1,2},$ $j=\overline{0,3}$ and all $n\in \mathbb{Z}_{+}.$
\end{theorem}

\begin{proof}
It is enough to observe that for any smooth solutions $b$ and $\tilde{b}:%
\mathcal{\tilde{M}}_{(3)}\rightarrow \mathbb{R}\{\{x,t\}\}$ to
functional-differential equations (\ref{R6}) and (\ref{R7}), respectively, \
the expressions $(\alpha D_{x})b$ and $(\alpha D_{x})\tilde{b}$ are
solutions to the determining functional-differential equation (\ref{R6}).
Iterating the operator $\alpha D_{x},$ one obtains the theorem statement.
\end{proof}

We proceed now to analyze the solution set to\ the functional-differential
equation (\ref{R3}), making use of the following transformation:
\begin{equation}
r:=\frac{a}{\alpha \eta },  \label{R10}
\end{equation}%
where $\eta :\mathcal{\tilde{M}}_{(3)}\rightarrow \mathbb{R}\{\{x,t\}\}$ is
any solution to equation (\ref{R9}) and a smooth functional mapping $a:$ $%
\mathcal{M}_{(3)}\rightarrow \mathbb{R}\{\{x,t\}\}$ satisfies the following
determining functional-differential equation:%
\begin{equation}
D_{t}a=\alpha \eta .  \label{R11}
\end{equation}%
Then any solution to functional-differential equation (\ref{R3}) has the
form
\begin{equation}
r=\frac{a}{\alpha \eta }+\eta _{0},  \label{R12}
\end{equation}%
where $\eta _{0}:\mathcal{\tilde{M}}_{(3)}\rightarrow \mathbb{R}\{\{x,t\}\}$
is any smooth solution to the functional-differential equation (\ref{R9}).

To find solutions to equation (\ref{R11}), we make use of the following
linear $\alpha $-expansion in the corresponding Riemann differential ideal $%
R\{\alpha\}$ $\subset \mathcal{K}\{\alpha \}:$
\begin{equation}
a=c_{3}+c_{0}\alpha +c_{1}\dot{\alpha}+c_{2}\ddot{\alpha}\in R\{\alpha \},%
\text{ \ \ }  \label{R13}
\end{equation}%
where $\dot{\alpha}:=D_{t}\alpha ,\ddot{\alpha}:=D_{t}^{2}\alpha $ and
taking into account that all functions $\alpha ,\dot{\alpha}$ and $\ddot{%
\alpha}$ are functionally independent owing to the fact that $\dddot{\alpha}%
:=D_{t}^{3}\alpha =0.$ As a result of substitution (\ref{R13}) into (\ref%
{R11}) we obtain the relationships%
\begin{equation}
\dot{c}_{1}+c_{0}=0,\text{ }\dot{c}_{0}=\eta ,\text{ }\dot{c}_{2}+c_{1}=0,%
\text{ }\dot{c}_{3}+c_{2}=0.  \label{R14}
\end{equation}%
Whence, owing to (\ref{R7}) we have at the special solution $\eta =1$ to
equation (\ref{R9}) two functional solutions for the mapping $c_{0}:\mathcal{%
\tilde{M}}_{(3)}\rightarrow \mathbb{R}\{\{x,t\}\}:$%
\begin{equation}
c_{0}^{(1)}=\frac{v}{z},\text{ \ \ }c_{0}^{(2)}=\frac{v_{x}}{z_{x}}.
\label{R14a}
\end{equation}%
As a result, solving the recurrent functional equations (\ref{R14}) yields

\begin{eqnarray}
a_{2}^{(1)} &=&[(xv-u^{2}/2)/z]_{x},\text{ }a_{2}=\frac{v_{x}}{z_{x}^{2}}-%
\frac{u_{x}^{2}}{2z_{x}^{2}},\text{ \ \ }  \label{R15} \\
a_{2}^{(1)} &=&\frac{v_{x}v^{3}}{6z_{x}z^{3}}-\frac{u_{x}v^{2}}{2z_{x}z^{2}}+%
\frac{u(uz-v^{2})}{6z^{3}}+\frac{v}{zz_{x}},  \notag
\end{eqnarray}%
giving rise to the following three functional solutions to (\ref{R3}):%
\begin{eqnarray}
r_{1}^{(1)} &=&\frac{v_{x}v^{3}}{6z^{3}}-\frac{u_{x}v^{2}}{2z^{2}}+\frac{%
u(uz-v^{2})z_{x}}{6z^{3}}+\frac{v}{z},  \label{R16} \\
r_{1}^{(2)} &=&(xv-u^{2}/2)/z]_{x},\text{ }r_{2}=\frac{v_{x}}{z_{x}}-\frac{%
u_{x}^{2}}{2z_{x}}.  \notag
\end{eqnarray}%
Having now chosen the next special solution $\eta :=b_{2}=\frac{u_{x}}{z_{x}}%
-\frac{v_{x}^{2}}{2z_{x}^{2}}$ to equation (\ref{R9}), one easily obtains
from (\ref{R14}) that the functional expression
\begin{equation}
r_{3}=(\frac{u_{x}^{3}}{6z_{x}^{2}}-\frac{u_{x}v_{x}}{2z_{x}^{2}}+\frac{3}{%
4z_{x}})/(\frac{u_{x}}{z_{x}}-\frac{v_{x}^{2}}{2z_{x}^{2}})  \label{R17}
\end{equation}%
also solves the functional-differential equation (\ref{R14}). Proceeding as
above, one can construct an infinite set ${\Omega }$ of the desired
solutions to the functional-differential equation (\ref{R14}) on the
manifold $\mathcal{\tilde{M}}_{(3)}.$ Thereby one has the following theorem.

\begin{theorem}
The complete set $\mathcal{R}$ of functional-differential solutions to
equation (\ref{R3}) on the manifold $\mathcal{\tilde{M}}_{(3)}$ is generated
by functional solutions in the form (\ref{R12}) to the reduced
functional-differential equations (\ref{R9}) and (\ref{R11}).
\end{theorem}

In particular, the subset
\begin{eqnarray}
\mathcal{\tilde{R}} &=&\{r_{1}^{(1)}=\frac{v_{x}v^{3}}{6z^{3}}-\frac{%
u_{x}v^{2}}{2z^{2}}+\frac{u(uz-v^{2})z_{x}}{6z^{3}}+\frac{v}{z},\text{ }%
r_{1}^{(2)}=[(xv-u^{2}/2)/z]_{x},  \label{R18} \\
\text{ }r_{2} &=&\frac{v_{x}}{z_{x}}-\frac{u_{x}^{2}}{2z_{x}},r_{3}=(\frac{%
u_{x}^{3}}{6z_{x}^{2}}-\frac{u_{x}v_{x}}{2z_{x}^{2}}+\frac{3}{4z_{x}})/(%
\frac{u_{x}}{z_{x}}-\frac{v_{x}^{2}}{2z_{x}^{2}})\}\subset \mathcal{R}
\notag
\end{eqnarray}%
coincides exactly with that found in \cite{GPPP,GBPPP,PoP}.

\subsection{The generalized Riemann type hydrodynamical equation: the case
N=4}

Now consider the generalized Riemann type differential equation (\ref{H0a})
at $N=4$
\begin{equation}
D_{t}^{4}u=0  \label{A10a}
\end{equation}%
on an element $u\in \mathbb{R}\{\{x,t\}\}$ and construct the related
invariant Riemann differential ideal $R\mathcal{\{}u\mathcal{\}}\subset
\mathcal{K}\{u\}$ as follows:%
\begin{eqnarray}
R\mathcal{\{}u\mathcal{\}} &:&=\{\lambda ^{3}\sum_{n\in \mathbb{Z}%
_{+}}f_{n}^{(1)}D_{x}^{n}u-\lambda ^{2}\sum_{n\in \mathbb{Z}%
_{+}}f_{n}^{(2)}D_{t}D_{x}^{n}u+\lambda \sum_{n\in \mathbb{Z}%
_{+}}f_{n}^{(3)}D_{t}^{2}D_{x}^{n}u-  \label{A10} \\
-\sum_{n\in \mathbb{Z}_{+}}f_{n}^{(4)}D_{t}^{3}D_{x}^{n}u
&:&D_{t}^{4}u=0,\lambda \in \mathbb{R},f_{n}^{(k)}\in \mathcal{K}\{u\},k=%
\overline{1,4},n\in \mathbb{Z}_{+}\}  \notag
\end{eqnarray}%
at a fixed function $u\in \mathbb{R}\{\{x,t\}\}.$ The Riemann differential
ideal (\ref{A10}), satisfying the $D_{t}$-invariance condition, is in this
case also maximal. The corresponding kernel $Ker$ $D_{t}\subset R\mathcal{\{}%
u\mathcal{\}}$ of the differentiation $D_{t}:\mathcal{K}\{u\}\rightarrow
\mathcal{K}\{u\},$ reduced upon the Riemann differential ideal (\ref{A10}),
is given by the following linear differential relationships:
\begin{equation}
D_{t}f^{(1)}=0,\text{ }D_{t}f^{(2)}=\lambda f^{(1)},\text{ }%
D_{t}f^{(3)}=\lambda f^{(2)},\text{ }D_{t}f^{(4)}=\lambda f^{(3)},
\label{A11}
\end{equation}%
where $f^{(k)}:=f^{(k)}(\lambda )=\sum_{n\in \mathbb{Z}_{+}}f_{n}^{(k)}%
\lambda ^{n}\in \mathcal{K}\{u\}|_{\mathcal{M}_{(4)}}=\mathbb{R}\{\{x,t\}\},$
$k=\overline{1,4}$ and $\lambda \in \mathbb{R}$ is arbitrary. The linear
relationships (\ref{A11}) \ can be easily represented in the space $\mathbb{R%
}^{4}\{\{x,t\}\}$ in the following matrix form:%
\begin{equation}
D_{t}f=q(\lambda )f,\text{ \ }q(\lambda ):=\left(
\begin{array}{cccc}
0 & 0 & 0 & 0 \\
\lambda & 0 & 0 & 0 \\
0 & \lambda & 0 & 0 \\
0 & 0 & \lambda & 0%
\end{array}%
\right) ,  \label{A12}
\end{equation}%
where $f:=(f^{(1)},f^{(2)},f^{(3)},f^{(4)})^{\intercal }\in \mathbb{R}%
^{4}\{\{x,t\}\},$ and $\lambda \in \mathbb{R}.$\ \ Moreover, it is easy to
observe that relationships (\ref{A11}) can be equivalently rewritten in the
compact scalar form as
\begin{equation}
D_{t}^{4}f^{(4)}=0,  \label{A13}
\end{equation}%
where an element $f_{4}\in \mathcal{K}\{u\}.$ Thus, now one can construct
the invariant Lax differential ideal, isomorphically equivalent to (\ref{A10}%
), as follows:
\begin{eqnarray}
L\mathcal{\{}u\mathcal{\}} &:&=%
\{g_{1}f^{(4)}+g_{2}D_{t}f^{(4)}+g_{3}D_{t}^{2}f^{(4)}+g_{4}D_{t}^{3}f^{(4)}:D_{t}^{4}f^{(4)}=0,
\notag \\
g_{j} &\in &\mathcal{K}\{u\},j=\overline{1,4}\}\subset \mathcal{K}\{u\},
\label{A14}
\end{eqnarray}%
whose $D_{x}$-invariance should be checked separately. The latter gives rise
to the representation

\begin{equation}
D_{x}f=\ell \lbrack u,v,w,z;\lambda ]f,\text{ \ }\ell \lbrack
u,v,w,z;\lambda ]:=\left(
\begin{array}{cccc}
-\lambda ^{3}u_{x} & \lambda ^{2}v_{x} & -\lambda w_{x} & z_{x} \\
-4\lambda ^{2} & 3\lambda ^{3}u_{x} & -2\lambda ^{2}v_{x} & \lambda w_{x} \\
-10\lambda ^{5}r_{1} & 6\lambda ^{4} & -3\lambda ^{3}u_{x} & \lambda
^{2}v_{x} \\
-20\lambda ^{6}r_{2} & 10\lambda ^{5}r_{1} & -4\lambda ^{4} & \lambda
^{3}u_{x}%
\end{array}%
\right) ,  \label{A15}
\end{equation}%
where we put, by definition,
\begin{equation}
D_{t}u:=v,D_{t}v:=w,D_{t}w:=z,D_{t}z:=0,  \label{A15a}
\end{equation}%
$(u,v,w,z)^{\intercal }\in \mathcal{\tilde{M}}_{(4)}\subset \mathbb{R}%
^{3}\{\{x,t\}\},$ and the mappings $r_{j}:\mathcal{\tilde{M}}%
_{(4)}\rightarrow \mathbb{R}\{\{x,t\}\},j=\overline{1,2},$ satisfy the
following functional-differential equations:%
\begin{equation}
D_{t}r_{1}+r_{1}D_{x}u=1,\text{ \ \ \ \ \ \ }D_{t}r_{2}+r_{2}D_{x}u=r_{1},
\label{A16}
\end{equation}%
similar to (\ref{A8}), considered above. The equations (\ref{A16}) possess
many different solutions, amongst which are the functional expressions:%
\begin{eqnarray}
r_{1} &=&D_{x}(\frac{uw^{2}}{2z^{2}}-\frac{vw^{3}}{3z^{3}}+\frac{vw^{4}}{%
24z^{4}}+\frac{7w^{5}}{120z^{4}}-\frac{w^{6}}{144z^{5}}),  \label{A17} \\
r_{2} &=&D_{x}(\frac{uw^{3}}{3z^{3}}-\frac{vw^{4}}{6z^{4}}+\frac{3w^{6}}{%
80z^{5}}+\frac{vw^{5}}{120z^{5}}-\frac{w^{7}}{420z^{6}}).  \notag
\end{eqnarray}%
Whence, we obtain the following proposition.

\begin{proposition}
The expressions (\ref{A12}) and (\ref{A15}) are the linear matrix
representations in the space $\mathbb{R}^{4}\mathbb{\{\{}x,t\}\}$ of the
differentiations $D_{t}:\mathcal{K}\{u\}\rightarrow \mathcal{K}\{u\}$ \ and $%
D_{x}:\mathcal{K}\{u\}\rightarrow \mathcal{K}\{u\},\ $respectively, reduced
upon the invariant Lax differential ideal $L\mathcal{\{}u\mathcal{\}}\subset
\mathcal{K}\{u\}$ given by (\ref{A6}).
\end{proposition}

Based now on the representations (\ref{A12}) and (\ref{A15}) one easily
constructs a standard Lax type representation, characterizing the
integrability of the nonlinear dynamical system
\begin{equation}
\left.
\begin{array}{c}
u_{t}=v-uu_{x} \\
v_{t}=w-uv_{x} \\
w_{t}=z-uw_{x} \\
z_{t}=-uz_{x}%
\end{array}%
\right\} :=K[u,v,w,z],  \label{A18}
\end{equation}%
equivalent to the generalized Riemann type hydrodynamical system (\ref{A10a}%
). Namely, the following theorem holds.

\begin{theorem}
The dynamical system (\ref{A18}), equivalent to the generalized Riemann type
hydrodynamical system (\ref{A10a}), possesses the Lax type representation%
\begin{equation}
\begin{array}{c}
f_{x}=\ell \lbrack u,v,z,w;\lambda ]{f},\text{ \ \ }f_{t}=p(\ell )f,\text{ \
\ }p(\ell ):=-u\ell \lbrack u,v,w,z;\lambda ]+q(\lambda ),%
\end{array}
\label{A19a}
\end{equation}%
where $f\in \mathbb{R}^{4}\{\{x,t\}\},$ $\lambda \in \mathbb{R}$ is a
spectral parameter and
\begin{equation*}
\ell \lbrack u,v,w,z;\lambda ]:=\left(
\begin{array}{cccc}
-\lambda ^{3}u_{x} & \lambda ^{2}v_{x} & -\lambda w_{x} & z_{x} \\
-4\lambda ^{4} & 3\lambda ^{3}u_{x} & -2\lambda ^{2}v_{x} & \lambda w_{x} \\
-10\lambda ^{5}r_{1} & 6\lambda ^{4} & -3\lambda ^{3}u_{x} & \lambda
^{2}v_{x} \\
-20\lambda ^{6}r_{2} & 10\lambda ^{5}r_{1} & -4\lambda ^{4} & \lambda
^{3}u_{x}%
\end{array}%
\right) ,\text{ \ }q(\lambda ):=\left(
\begin{array}{cccc}
0 & 0 & 0 & 0 \\
\lambda & 0 & 0 & 0 \\
0 & \lambda & 0 & 0 \\
0 & 0 & \lambda & 0%
\end{array}%
\right) ,
\end{equation*}%
\begin{equation}
\text{ \ }p(\ell )=\left(
\begin{array}{cccc}
\lambda uu_{x} & -\lambda ^{2}uv_{x} & \lambda uw_{x} & -uz_{x} \\
\lambda +4\lambda ^{4}u & -3\lambda ^{3}uu_{x} & 2\lambda ^{2}uv_{x} &
-\lambda uw_{x} \\
10\lambda ^{5}ur_{1} & \lambda -6\lambda ^{4}u & 3\lambda ^{3}uu_{x} &
-\lambda ^{2}uv_{x} \\
20\lambda ^{6}ur_{2} & -10\lambda ^{5}ur_{1} & \lambda +4\lambda ^{4}u &
-\lambda ^{3}uu_{x}%
\end{array}%
\right) ,  \label{A19}
\end{equation}%
so it is a Lax type integrable dynamical system on the functional manifold $%
\mathcal{\tilde{M}}_{(4)}.$
\end{theorem}

The result obtained above can be easily generalized on the case of an
arbitrary integer $N\in \mathbb{Z}_{+},$ thereby proving the Lax type
integrability of the whole hierarchy of the Riemann type hydrodynamical
equation (\ref{H0a}). The related calculations will be presented and
discussed in other work. Here we only do the next remark.

\begin{remark}
The Riemann type hydrodynamical equation (\ref{H0a}) as $N\rightarrow \infty
$ can be equivalently rewritten as the following Benney type \cite%
{Ben,KM,PBB} chain
\begin{equation}
D_{t}u^{(n)}=u^{(n+1)},\text{ \ \ }D_{t}:=\partial /\partial
t+u^{(0)}\partial /\partial x,  \label{A20}
\end{equation}%
for the suitably constructed moment functions $%
u^{(n)}:=D_{t}^{n}u^{(0)},u^{(0)}:=u\in \mathbb{R}\{\{x,t\}\},$ $n\in
\mathbb{Z}_{+}.$
\end{remark}

This aspect of the problem is very interesting and we plan to treat it in
detail by means of the differential-geometric tools elsewhere.

\bigskip

\section{\protect\bigskip The differential-algebraic analysis of the Lax
type integrability of the Korteweg-de Vries dynamical system}

\subsection{The differential-algebraic problem setting}

We consider the well known Korteweg-de Vries equation in the following (\ref%
{A1-2}) differential-algebraic form:

\begin{equation}
D_{t}u-D_{x}^{3}u=0,  \label{K1}
\end{equation}%
where $u\in \mathcal{K}\{u\}$ and the differentiations $D_{t}:=\partial
/\partial t+u\partial /\partial x,D_{x}:=\partial /\partial x$ satisfy the
commutation condition (\ref{A1}):
\begin{equation}
\lbrack D_{x},D_{t}]=(D_{x}u)D_{x}.  \label{K2}
\end{equation}%
We will also interpret relationship (\ref{K1}) as a nonlinear dynamical
system
\begin{equation}
D_{t}u=D_{xxx}u  \label{K1a}
\end{equation}%
on a suitably chosen functional manifold $\mathcal{M}$ $\subset $ $\mathbb{R}%
\{\{x,t\}\}.$

Based on the expression (\ref{K1}) we can easily construct a suitable
invariant KdV-differential ideal $KdV\mathcal{\{}u\mathcal{\}}\subset
\mathcal{K}\{u\}$ as follows:%
\begin{eqnarray}
KdV\mathcal{\{}u\mathcal{\}} &:&=\{\sum_{k=\overline{0,2}}\sum_{n\in \mathbb{%
Z}_{+}}f_{n}^{(k)}D_{x}^{k}D_{t}^{n}u\in \mathcal{K}\{u\}:\text{ }%
D_{t}u-D_{x}^{3}u=0,  \notag \\
f_{n}^{(k)} &\in &\mathcal{K}\{u\},k=\overline{0,2},n\in \mathbb{Z}%
_{+}\}\subset \mathcal{K}\{u\}.  \label{K3}
\end{eqnarray}%
The ideal (\ref{K3}) proves to be not maximal, that seriously influences on
the form of the reduced modulo it representations of derivatives $\ D_{x}$
and $D_{t}:\mathcal{K}\{u\}\rightarrow \mathcal{K}\{u\}.$ \ As the next step
we need to find the kernel $Ker$ $D_{t}\subset KdV\mathcal{\{}u\mathcal{\}}$
of the differentiation $D_{t}:\mathcal{K}\{u\}\rightarrow \mathcal{K}\{u\},$
reduced upon the KdV-differential ideal (\ref{K3}). \ We obtain by means of
easy calculations that it is generated by the following differential
relationships:%
\begin{eqnarray}
D_{t}f^{(0)} &=&-\lambda f^{(0)},\text{ \ }D_{t}f^{(2)}=-\lambda f^{(2)}%
\text{\ }+2f^{(2)}D_{x}u,  \notag \\
D_{t}f^{(1)} &=&-\lambda f^{(1)}\text{\ }+f^{(1)}D_{x}u+f^{(2)}D_{xx}u,
\label{K4}
\end{eqnarray}%
2where, by definition, $\ f^{(k)}:=f^{(k)}(\lambda )=\sum_{n\in \mathbb{Z}%
_{+}}f_{n}^{(k)}\lambda ^{n}\in \mathcal{K}\{u\}|_{\mathcal{M}}=\mathbb{R}%
\{\{x,t\}\},k=\overline{0,2},$ and $\lambda \in \mathbb{R}$ $\ $ is an
arbitrary parameter. Based on the relationships (\ref{K4}) the following
proposition holds.

\begin{proposition}
The differential relationships (\ref{K4}) can be equivalently rewritten in
the following linear matrix form:%
\begin{equation}
D_{t}f=q(\lambda )f,\text{ \ }q(\lambda ):=\left(
\begin{array}{cc}
D_{x}u-\lambda  & D_{xx}u \\
0 & 2D_{x}u-\lambda
\end{array}%
\right) ,  \label{K5}
\end{equation}%
where\ $f:=(f_{1},f_{2})^{\intercal }\in \mathbb{R}^{2}{\{\{}x,t\}\},$ $%
\lambda \in \mathbb{R},$ giving rise to the corresponding linear matrix
representation in the space $\mathbb{R}^{2}\{\{x,t\}\}$ of the
differentiation $D_{t}:\mathcal{K}\{u\}\rightarrow \mathcal{K}\{u\},$
reduced upon the KdV-differential ideal (\ref{K3}).
\end{proposition}

\subsection{The Lax type representation}

Now, making use of the matrix differential relationship (\ref{K5}), we can
construct the Lax differential ideal related to the ideal (\ref{K3})
\begin{eqnarray}
L\mathcal{\{}u\mathcal{\}} &:&=\{<g,f>_{\mathbb{E}^{2}}\in \mathcal{K}%
\{u\}:D_{t}f=q(\lambda )f,\text{ }  \notag \\
\text{\ }f,g &\in &\mathcal{K}^{2}\{u\}\text{ }\}\subset \mathcal{K}\{u\},
\label{K6}
\end{eqnarray}%
where $<\cdot ,\cdot >_{\mathbb{E}^{2}}$ denotes the standard scalar product
in the Euclidean real space $\ \mathbb{E}^{2}.$ Since the Lax differential
ideal (\ref{K6}) is, by construction, $D_{t}$-invariant and isomorphic to
the $D_{t}$- and $D_{x}$-invariant KdV-differential ideal (\ref{K3}), it is
necessary to check its $D_{x}$-invariance. As a result of this condition the
following differential relationship
\begin{equation}
D_{x}f=\ell \lbrack u;\lambda ]f,\text{ }\ell \lbrack u;\lambda ]:=\left(
\begin{array}{cc}
D_{x}\tilde{a} & 2D_{xx}\tilde{a}. \\
-1 & -D_{x}\tilde{a}%
\end{array}%
\right)  \label{K7}
\end{equation}%
holds, where the mapping $\tilde{a}:\mathcal{M}\rightarrow \mathbb{R}%
\{\{x,t\}\}$ satisfies the functional-differential relationships
\begin{equation}
D_{t}\tilde{a}=1,D_{t}u-D_{x}^{3}u=0,  \label{K7a}
\end{equation}%
and the matrix $\ell :=\ell \lbrack u;\lambda ]:\mathbb{R}%
^{2}\{\{x,t\}\}\rightarrow \mathbb{R}^{2}\{\{x,t\}\}$ satisfies for all $%
\lambda \in \mathbb{R}$ the determining functional-differential equation%
\begin{equation}
D_{t}\ell +\ell D_{x}u=[q(\lambda ),\ell ]+D_{x}q(\lambda ),  \label{K8}
\end{equation}%
generalizing the similar equation (\ref{A8a}). The result obtained above we
formulate as the following proposition.

\begin{theorem}
The derivatives $D_{t}:\mathbb{R}\{\{x,t\}\}\rightarrow \mathbb{R}\{\{x,t\}\}
$ and $D_{x}:\mathcal{K}\{u\}\rightarrow \mathcal{K}\{u\}$ of the
differential ring $\mathcal{K}\{u\},$ reduced upon the Lax differential
ideal $L\mathcal{\{}u\mathcal{\}}\subset \mathcal{K}\{u\},$ which isomorphic
to the KdV-differential ideal $KdV\mathcal{\{}u\mathcal{\}}\subset \mathcal{K%
}\{u\},$ \ allow the compatible Lax type representation (generated by the
invariant Lax differential ideal $L\mathcal{\{}u\mathcal{\}}\subset \mathcal{%
K}\{u\}$)
\begin{eqnarray}
D_{t}f &=&q(\lambda )f,\text{ \ }q(\lambda ):=\left(
\begin{array}{cc}
D_{x}u-\lambda  & D_{xx}u \\
0 & 2D_{x}u-\lambda
\end{array}%
\right) ,  \notag \\
D_{x}f &=&\ell \lbrack u;\lambda ]f,\text{ }\ell \lbrack u;\lambda ]:=\left(
\begin{array}{cc}
D_{x}\tilde{a} & 2D_{xx}\tilde{a}. \\
-1 & -D_{x}\tilde{a}%
\end{array}%
\right) ,  \label{K9}
\end{eqnarray}%
where the mapping $\tilde{a}:\mathcal{M}\rightarrow \mathbb{R}\{\{x,t\}\}$
satisfies the functional-differential relationships (\ref{K7a}), $f\in
\mathbb{R}^{2}\{\{x,t\}\}\ $and $\lambda \in \mathbb{R}.$
\end{theorem}

It is interesting to mention that the Lax type representation (\ref{K9})
strongly differs from that given by the well known \cite{Nov} classical
expressions%
\begin{eqnarray}
D_{t}f &=&q_{cl}(\lambda )f,\text{ \ }q_{cl}(\lambda ):=\left(
\begin{array}{cc}
D_{x}u/6 & -(2u/3-4\lambda ) \\
\begin{array}{c}
D_{xx}u/6-(u/6-\lambda )\times \\
\times (2u/3-4\lambda )%
\end{array}
& -11D_{x}u/6%
\end{array}%
\right) ,  \notag \\
D_{x}f &=&\ell _{cl}[u;\lambda ]f,\text{ }\ell _{cl}[u;\lambda ]:=\left(
\begin{array}{cc}
0 & 1 \\
u/6-\lambda & 0%
\end{array}%
\right) ,  \label{K10}
\end{eqnarray}%
where, as above, the following functional-differential equation (equivalent
to the nonlinear dynamical system (\ref{K1a}) on the functional manifold $%
\mathcal{M)}$
\begin{equation}
D_{t}\ell _{cl}+\ell _{cl}D_{x}u=[q_{cl}(\lambda ),\ell
_{cl}]+D_{x}q_{cl}(\lambda ),  \label{K11}
\end{equation}%
holds for $\ $any $\lambda \in \mathbb{R}.$ This fact, as we suspect, is
related with the existence of different $D_{t}$-invarinat KdV-differential
ideals of form (\ref{K3}), which are not maximal. Thus, a problem of
constructing a suitable KdV-differential ideal $KdV\mathcal{\{}u\mathcal{\}}$
$\subset \mathcal{K}\{u\}$ generating the corresponding invariant Lax type
differential ideal $L\mathcal{\{}u\mathcal{\}}\subset \mathcal{K}\{u\},$
invariant with respect to the differential representations (\ref{K10}),
naturally arises, and we expect to treat this in detail elsewhere. There
also is a very interesting problem of the differential-algebraic analysis of
the related symplectic structures on the functional manifold $\mathcal{M},$
with respect to which the dynamical system (\ref{K1a}) is Hamiltonian and
suitably integrable. Here we need also to mention a very interesting work
\cite{Wi}, where the integrability structure of the Korteweg-de Vries
equation was analyzed from the differential-algebraic point of view.

\section{Conclusion}

The results presented provide convincing evidence that the
differential-algebraic tools, when applied to a given set of differential
relationships based on the derivatives $D_{t}$ and $D_{x}:\mathcal{K}%
\{u\}\rightarrow \mathcal{K}\{u\}$ in the differential ring $\mathcal{K}%
\{u\} $ and parameterized by a fixed element $u\in \mathcal{K}\{u\},$ make
it possible to construct the corresponding Lax type representation as that
realizing \ the linear matrix representations of the derivatives reduced
modulo the corresponding invariant Riemann differential ideal. This scheme
was elaborated in detail for the generalized Riemann type differential
equation (\ref{H0a}) and for the classical Korteweg-de Vries equation (\ref%
{K1a}). As these equations are equivalent to the corresponding Hamiltonian
systems with respect to suitable symplectic structures, this aspect presents
a very interesting problem from the differential-algebraic point of view,
which we plan to study in the near future.\bigskip

\section{Acknowledgements}

The authors are much obliged to Profs. Zbignew Peradzy\'{n}ski (Warsaw
University, Poland), Jan S\l awianowski (IPPT, Poland) and Denis Blackmore
(NJIT, New Jersey, USA) for very valuable discussions. A.K.P. is cordially
thankful to Dr. Camilla Hollanti (Turku University, Finland) for the
invitation to take part in the "Cohomology Course-2010", to deliver a report
and for the nice hospitality. M.V.P. was in part supported by RFBR grant
08-01-00054 and a grant of the RAS Presidium "Fundamental \ Problems in
Nonlinear Dynamics". He also thanks the Wroclaw University for the
hospitality. The last but not least thanks go to Referees, whose many
instrumental comments, suggestions and remarks made the exposition strongly
improved.

\end{document}